\documentclass{article}
\usepackage[utf8]{inputenc}
\usepackage{amsthm}
\usepackage{amssymb}
\usepackage{amsmath}
\usepackage{hyperref}
\usepackage{bbold}
\usepackage{mathtools} 
\usepackage[ruled, linesnumbered]{algorithm2e}
\usepackage[margin=1.1in]{geometry}

\usepackage[sorting=none, url=false, isbn=false]{biblatex}
\newtheorem{theorem}{Theorem}
\newtheorem{corollary}{Corollary}[theorem]
\newtheorem{lemma}[theorem]{Lemma}
\newtheorem{proposition}[theorem]{Proposition}
\newtheorem*{lemma*}{Lemma}
\newtheorem*{proposition*}{Proposition}
\DeclareMathOperator{\diff}{d \!}
\DeclareMathOperator\Tr{Tr}
\DeclareMathOperator\id{id}

\providecommand{\od}[3][]{\ensuremath{
\ifinner
\tfrac{\diff{^{#1}}#2}{\diff{{#3}^{#1}}}
\else
\dfrac{\diff{^{#1}}#2}{\diff{{#3}^{#1}}}
\fi
}}

\providecommand{\dod}[3][]{\ensuremath{\mathinner{
\dfrac{\diff{^{#1}}#2}{\diff{{#3}^{#1}}}
}}}

\newcommand\ket[1]{|#1\rangle}
\newcommand\bra[1]{\langle#1|}
\newcommand\ketbra[2]{|#1\rangle\langle#2|}
\newcommand\norm[1]{\left\lVert#1\right\rVert}

\newcommand{\setbuilder}[3][\null]{%
  \ifx#1\null
       { \left\{ #2 \;\middle|\; #3 \right\} }%
    \else%
       { #1\{ #2 \;#1|\; #3 #1\} }%
    \fi}

\newcommand{\defeq}{\coloneqq}

\title{Eigenpath traversal by Poisson-distributed phase randomisation}
\author{Joseph Cunningham\footnote{Centre for Quantum Information and Communication (QuIC), \'{E}cole polytechnique de Bruxelles,
Universit\'e libre de Bruxelles} \and Jérémie Roland\footnotemark[1]}
\date{}

\begin{document}
\maketitle 

\begin{abstract}
We present a framework for quantum computation, similar to Adiabatic Quantum Computation (AQC), that is  based on the quantum Zeno effect. By performing randomised dephasing operations at intervals determined by a Poisson process, we are able to track the eigenspace associated to a particular eigenvalue.

We derive a simple differential equation for the fidelity, leading to general theorems bounding the time complexity of a whole class of algorithms. We also use eigenstate filtering to optimise the scaling of the complexity in the error tolerance $\epsilon$.

In many cases the bounds given by our general theorems are optimal, giving a time complexity of  $O(1/\Delta_m)$ with $\Delta_m$ the minimum of the gap. This allows us to prove optimal results using very general features of problems, minimising the problem-specific insight necessary.

As two applications of our framework, we obtain optimal scaling for the Grover problem (i.e.\ $O(\sqrt{N})$ where $N$ is the database size) and the Quantum Linear System Problem (i.e.\ $O(\kappa\log(1/\epsilon))$ where $\kappa$ is the condition number and $\epsilon$ the error tolerance) by direct applications of our theorems.
\end{abstract}

\section{Introduction}
It has long been appreciated that the ability to prepare a ground state of a given Hamiltonian is useful for a large number computational tasks. Many NP-hard problems, including various types of partitioning, covering, and satisfiability problems, can be solved by finding the ground state of an Ising system \cite{lucasIsing}. There are also many applications in the fields of quantum chemistry, where finding the ground state of molecules is a common task, and physics, where knowledge of the ground state helps to understand low-temperature phenomena such as superconductivity and superfluidity.

For computational problems we have the following strategy: (1) find a physical system such that the ground state encodes useful information for solving the problem,  (2) prepare the ground state using some physical process and (3) use the information contained in the ground state to solve the problem. This paper is about performing the second step of this strategy.

The most famous way to perform the second step is known as Adiabatic Quantum Computation (AQC) \cite{farhiAQC}. Suppose $H_P$ is the Hamiltonian whose ground state is of interest. This procedure requires a second Hamiltonian, $H_0$, with an easily preparable ground state.  Now consider the following interpolated Hamiltonian: $H(s) = (1-s)H_0 + sH_P$ and pick a large time $T$. We start with the system in the ground state of $H_0$ and evolve according to the time-dependent Hamiltonian $H(t/T)$, for time $t\in [0,T]$. The adiabatic theorem says that if $T$ is large enough, then the resulting state will be close to the ground state of $H(1) = H_P$. See \cite{jansenAdiabaticTheorem} for results detailing how large $T$ has to be. Clearly we want to take $T$ as small as possible, since a larger $T$ means our computation takes longer.

While AQC is polynomially equivalent to the quantum circuit model \cite{AQCequivalent}, it suffers from a few drawbacks. The most significant one being that it requires the system to evolve under a very specific time-dependent Hamiltonian. It is typically very hard to physically implement a system that evolves under exactly this Hamiltonian. Often the complicated time-dependent dynamics are approximated by a sequence of simpler evolutions, which introduces discretisation error. In particular this is necessary when implementing AQC on a conventional quantum computer. Bounding the discretisation error analytically is typically hard to do. In contrast, our method only requires the evolution under a finite number of time-independent Hamiltonians for finite time and thus has no discretisation cost.

There exist alternatives to AQC that are also based on an interpolation $H(s)$ between a Hamiltonian whose ground state is easy to prepare and one whose ground state is difficult to prepare. These approaches use alternate ways to transform the ground state of $H(0)$ into that of $H(1)$, or, more generally some eigenstate of $H(0)$ into the corresponding eigenstate of $H(1)$.
They often make use of a variant of the quantum Zeno effect. For instance \cite{qSearchMeasurement} uses measurement and \cite{eigenpathTraversalPR} simulates the quantum Zeno effect by applying Hamiltonian evolutions for random amounts of time in a procedure known as the Randomisation Method (RM).

Our framework builds on the RM of \cite{eigenpathTraversalPR} in the following way: instead of performing a fixed sequence of phase randomising steps, we stochastically choose when to perform phase randomisation, based on a Poisson process with rate $\lambda(s)$, see algorithm \ref{procedure}.
This has a number of advantages. Firstly it yields a simple differential equation for the state evolution, which fits in the general framework of non-unitary adiabatic theorems of \cite{adiabaticTheoremsContractingEvolutions}, and greatly simplifies the analysis. It allows us to obtain general theorems, see in particular theorems \ref{theorem:constantRate} and \ref{theorem:adaptiveRate}, that in many cases yield optimal results with minimal extra work or problem-specific insight.
Also, we only need very minimal technical assumptions on $H(s)$: it only needs to be twice continuously differentiable and we need to know some estimate of the gap between the eigenvalue of interest and the rest of the spectrum. We do not assume precise knowledge of the spectrum or gap. We allow the eigenspace of interest to be degenerate.

Our theorem \ref{theorem:constantRate} deals with the case where the rate of the Poisson process $\lambda$ is taken to be constant. The result we obtain is better than the corresponding result for AQC with a constant-speed linear interpolation. In theorem \ref{theorem:adaptiveRate} we describe a variable $\lambda(s)$ that can significantly improve the time complexity, up to $O(1/\Delta_m)$ in the minimum gap $\Delta_m$. Finally theorem \ref{eigenstateFiltering} improves the dependence of the time complexity on the error tolerance. Typically algorithms based on AQC and the RM have a complexity that scales as $O(1/\epsilon)$ in the error tolerance $\epsilon$. Eigenstate filtering, introduced in \cite{eigenstateFiltering}, can be used to reduce this to $O(\log(1/\epsilon))$. This has been used before to RM-inspired algorithms that use the circuit model, see \cite{eigenstateFiltering} and \cite{QLSPwithPRnumerics}, but we provide a version native to our cost model.

From theorem \ref{theorem:adaptiveRate}, we see that the following property is very useful to obtain fast algorithms: $\int_0^1\frac{1}{\Delta(s)^p}\diff{s} = O(\Delta_m^{1-p})$, where $\Delta$ is the gap, $\Delta_m = \inf_{s\in[0,1]}\Delta$ and $p>1$. This property seems to be quite generic, in particular it holds for both the Grover search problem and the Quantum Linear System Problem (QLSP).

In the Grover search problem, \cite{Grover}, the goal is to prepare a specific state in an $N$-dimensional space with the help of an oracle. It is well-known that this can be done in $O(\sqrt{N})$ queries to the oracle. When AQC was first used to tackle this problem, a complexity of $O(N)$ was obtained \cite{farhiAQC}. The trick to achieving a complexity of $O(\sqrt{N})$ was to use an adapted schedule \cite{RolandLocalAQC}, \cite{otherLocalAQC}. In our framework, the algorithm using constant $\lambda$ already achieves a scaling of $O\big(\sqrt{N}\log(N)\big)$ by theorem \ref{theorem:constantRate}, which is significantly better than the corresponding case for AQC. Using a variable $\lambda(s)$, we recover the scaling $O(\sqrt{N})$ by theorem \ref{theorem:adaptiveRate}.
It is interesting to note that, by the generality of our theorems, we actually obtain a whole family of schedules, parametrised by some value $0<q<1$, that solve the problem optimally. This is analogous to the range of adiabatic schedules considered in \cite{An_QLSP} and \cite{QLSPdiscreteAdiabaticTheorem}. The original schedule of \cite{RolandLocalAQC} actually corresponds to a choice of $q=1$. It seems like the RM can be considered as the $q=1$ case of a family of methods, at least in the case of linear interpolation. This falls outside the range of our theorem, but it turns out that $q=1$ is in fact good enough to give optimality for the Grover problem, which explains why the RM was already known to be able to perform Grover search with optimal complexity $O(\sqrt{N})$, \cite{eigenpathTraversalPR}.

In general, for other problems, $q=1$ does not give optimal scaling.
The Quantum Linear System Problems (QSLP) is an example of a problem where $q=1$ does not work, neither in our framework, nor in AQC \cite{An_QLSP}.
In QLSP, \cite{HHL}, the goal is to prepare a quantum state $\ket{x}$ that is proportional to the solution of a system of linear equations $Ax = b$. In \cite{QLSPwithPR} the randomisation method was used to construct an algorithm with complexity $O\big(\kappa \log(\kappa)/\epsilon\big)$, where $\kappa$ is the condition number and $\epsilon$ the error tolerance. This is improved to $O\big(\kappa \log(\kappa/\epsilon)\big)$ in \cite{QLSPwithPRnumerics}. In \cite{QLSPdiscreteAdiabaticTheorem} an algorithm based on a discrete adiabatic theorem was proposed which scales as $O\big(\kappa\log(1/\epsilon)\big)$. This is known to be optimal \cite{HHL}. We are able to match this in our framework.

There has actually been some discussion recently on the merits of these two approaches to QLSP, i.e.\ the approach based on the RM of \cite{QLSPwithPR} and the approach based on the discrete adiabatic theorem \cite{QLSPdiscreteAdiabaticTheorem}. The approach based on the discrete adiabatic theorem has the better asymptotic scaling, but it turns out that the proven complexity for reasonable values of $\kappa$ is very large. The paper \cite{QLSPwithPRnumerics} presents an algorithm that is based on the RM and has a better proven complexity for reasonable values of $\kappa$, but is asymptotically suboptimal. Finally \cite{QLSPdiscreteAdiabaticTheoremNumerics} uses numerical methods to determine the actual performance of the algorithm based on the discrete adiabatic theorem. They claim that it works much better than the proven bound and in fact better than the algorithm based on the RM.
We can contribute to this discussion by noting that our framework gives an algorithm that is based on the RM and has optimal asymptotic scaling. In addition, since the RM seems to correspond to $q=1$, which we know to be suboptimal, it is likely that the algorithm of \cite{QLSPwithPRnumerics} can be made asymptotically optimal by changing the scheduling.

\subsection{General setup}
We assume we have a physical system and a set of (time-independent) Hamiltonians, i.e.\ self-adjoint operators, such that we can evolve the system under $e^{-itH}$ at a cost of $t$ for any Hamiltonian $H$ in this set.\footnote{We set $\hbar = 1$.} We call the Hamiltonians in this set admissible. Which Hamiltonians are admissible will depend on the device or setup, but typically they will be bounded in norm.

This is not the cost model used by references \cite{QLSPdiscreteAdiabaticTheorem} and \cite{QLSPwithPRnumerics}, which use a query complexity rather than a time complexity. We discuss a translation of our results to this model using optimal Hamiltonian simulation in appendix \ref{appendix:circuitModel}. The asymptotic complexities are mostly unaffected, but there are different constants involved.

For a given instance of a problem, we assume that we have a continuous, twice differentiable path of admissible Hamiltonians $H(s)$, where $s\in [0,1]$. We also assume that we can prepare the ground state of $H(0)$.

We are interested in the asymptotic scaling of the time complexity, as measured by the total length of time we apply unitaries of the form $e^{-itH}$. We produce theorems that give bounds on the complexity in terms of the spectral gap and the derivatives $\norm{H'}$ and $\norm{H^{\prime\prime}}$.

Our main tool will be the randomised application of unitaries of the form $e^{-itH}$. Since we are using classical randomness, it will be useful to use the density matrix formalism. Using this formalism, we can derive differential equations for these new procedures that share essential features with the Liouville–von Neumann equation in the adiabatic limit. This allows us to use many of the same mathematical tricks to study these procedures and we can derive ``generalised'' adiabatic theorems in the sense of \cite{adiabaticTheoremsContractingEvolutions}.

\subsubsection{Cost and error model}
It is clear what the cost of one run of the algorithm is: it is just the total time spent evolving the system under some Hamiltonian. In order to state the time complexity we have the additional problem that the running time of the algorithm is not deterministic. That is, even for a fixed input, multiple runs of the algorithm will take different amounts of time. Our time complexity uses the expected run time of each input. Thus we say our algorithm has time complexity $T$ if, for all relevant inputs $I$, the expected time taken by the algorithm with input $I$ is less than $T$. In other words, we may consider this a worst-case expected-time complexity.

In order to guarantee that the algorithm does not take too long, we could abort if the chosen amount of time was too long. This would yield an additional error, which can be bounded by Markov's inequality.

Our algorithms are also not guaranteed to give the correct answer, rather we aim to produce the target state with at least a certain target fidelity.

\subsubsection{Technical assumptions on the spectrum}
\label{assumptions}

We assume the existence of the following objects: a number $\Delta_m >0$ and functions $\omega_0: [0,1]\to \mathbb{R}$ and $\Delta: [0,1]\to [0,1]$ such that
\begin{itemize}
\item $\omega_0$ continuous;
\item $\omega_0(s)$ is an eigenvalue of $H(s)$ for all $s\in [0,1]$;
\item $\Delta(s) \geq \Delta_m$ for all $s\in [0,1]$;
\item the intersection of $[\omega_0(s) - \Delta(s), \omega_0(s) + \Delta(s)]$ with the spectrum of $H(s)$ is exactly $\{\omega_0(s)\}$.
\end{itemize}
Let $P(s)$ be the projector on the eigenspace associated to the eigenvalue $\omega(s)$. We also set $Q(s) = \id - P(s)$.

In order to perform our algorithm, we assume knowledge of $\Delta$, which bounds the gap. We do not assume more detailed knowledge of the gap, $\omega_0$, or any other part of the spectrum.

\section{Poisson-distributed phase randomisation}
Our algorithms are built using a finite number of steps, where at each step a Zeno-like dephasing operation is performed. This dephasing operation is given by the following proposition:

\begin{proposition}[Phase randomisation \cite{eigenpathTraversalPR}] \label{phaseRandomisation}
Let $H$ be a Hamiltonian and $\omega_0, \Delta, P$ and $Q$ as above. Assume we can simulate $e^{-itH}$ for any positive or negative time $t$ at a cost of $|t|$. Then we can construct a stochastic variable $\tau$ such that for all states $\rho$,
\begin{equation}
\langle e^{-i\tau H}\rho e^{i\tau H}\rangle = P\rho P + Q\langle e^{-i\tau H}\rho e^{i\tau H} \rangle Q,
\end{equation}
with cost $\langle |\tau|\rangle = t_0/\Delta$, where $t_0 = 2.32132$. 
\end{proposition}
The angled brackets mean taking the average over $\tau$. The result is originally from \cite{eigenpathTraversalPR}. The value for $t_0$ was obtained in theorem 2 of \cite{QLSPwithPRnumerics}.

The algorithm is now simple to state:
\begin{algorithm}
Pick a Poisson process $N: [0,1] \times (\Omega, \mathcal{A}, P)\to \mathbb{N}$ with rate $\lambda(s)$\;
At each jump point $s$ of the Poisson process, pick an instance $t$ of the random variable $T$ as defined in proposition \ref{phaseRandomisation} and evolve the system under the Hamiltonian evolution $e^{-itH(s)}$\;
\caption{Poisson-distributed phase randomisation} \label{procedure}
\end{algorithm}

The density matrix describing the system is a random variable that satisfies the stochastic differential equation $\diff{\rho} = \big(e^{-i\tau(s)H(s)}\rho e^{i\tau(s)H(s)}- \rho\big)\diff{N}.$
Averaging over realisations, we get
\begin{equation} \diff{\langle\rho\rangle} = \big(P\langle\rho\rangle P + Q\langle e^{-i\tau H}\rho e^{i\tau H} \rangle Q - \langle\rho\rangle\big)\lambda\diff{s}. \label{eq:diffEq} \end{equation}
Note that this should properly be thought of as a ``marginalised'' density matrix, rather than an ``average'' density matrix. This is entirely analoguous to the situation for classical probability distributions, where integrating out a variable gives the marginal distribution.
In this case we are marginalising over the choice of Poisson process $N$. In the rest of the paper, we will use $\rho$ to refer to the marginalised distribution $\langle\rho\rangle$. This corresponds to the density matrix you would observe if you were not told which Poisson process $N$ was chosen.

The total time taken by one run of the algorithm is a random variable $T$ satisfying
\begin{equation} \diff{T} = \tau\diff{N}. \end{equation}
In order to find the time complexity, we take the average. This gives $\diff{T} = \Delta^{-1}\lambda\diff{s}$, so $T = \int_0^1\frac{\lambda}{\Delta}\diff{s}$.

\subsection{Analysis}
\begin{lemma} \label{lemma:errorBound}
Under the assumptions in \ref{assumptions}, the algorithm \ref{procedure} with rate $\lambda(s)$ produces a state with an infidelity that is bounded by
\begin{equation}
\epsilon \leq \norm{\lambda(0)^{-1}P'(0)} + \norm{\lambda(1)^{-1}P'(1)} + \int_0^1\bigg(\norm{\frac{P^{\prime\prime}}{\lambda}} + \Big|\Big(\frac{1}{\lambda}\Big)'\Big|\norm{P^{\prime}}\bigg)\diff{s}.
\end{equation}
\end{lemma}
\begin{proof}
The infidelity is given by $\epsilon = 1 - \Tr\big(P(1)\rho(1)\big) = \Tr\big(P(0)\rho(0)\big) - \Tr\big(P(1)\rho(1)\big) = \Big|\Tr(P\rho)\big|_0^1\Big|$, so it makes sense to track how the fidelity $\Tr\big(P(s)\rho(s)\big)$ changes in time.

We construct a differential equation for $\Tr(P\rho)$ by taking the derivative with respect to $s$, $\Tr(P\rho)' = \Tr(P'\rho) + \Tr(P\rho')$. This can be simplified using the fact that $PP'P = 0$ and $QP'Q = 0$.\footnote{We have $P' = (PP)' = P'P + PP'$, so $PP'P = 2PP'P$ and $QP'Q = 0$.} Indeed, we have
\begin{equation}
\Tr(P\rho') = \lambda \Tr\big(P(P\rho P + Q\langle e^{-i\tau H}\rho e^{i\tau H} \rangle Q  - \rho)\big) = \Tr(P\rho P) - \Tr(P\rho) = 0 \label{eq:firstPartPrho}
\end{equation}
and
\begin{align}
\Tr(P'\rho) &= \Tr\Big(P'(P\rho P + Q\langle e^{-i\tau H}\rho e^{i\tau H} \rangle Q  - \lambda^{-1}\rho')\Big) \\
&= \Tr(PP'P\rho) + \Tr\big((QP'Q)\langle e^{-i\tau H}\rho e^{i\tau H} \rangle\big) - \Tr(\lambda^{-1}P'\rho') \\
&= - \Tr(\lambda^{-1}P'\rho'), \label{eq:secondPartPrho}
\end{align}
so $\Tr(P\rho)' = - \Tr(\lambda^{-1}P'\rho')$. Integrating gives
\begin{align}
\Tr(P\rho)\big|_0^1 &= - \int_0^1\Tr(\lambda^{-1}P'\rho')\diff{s} \\
&= -\lambda^{-1}\Tr(P'\rho)\big|_0^1 + \int_0^1\bigg(\Tr\Big(\frac{P^{\prime\prime}}{\lambda}\rho\Big) + \Tr\Big(\big(\tfrac{1}{\lambda}\big)'P^{\prime}\rho\Big)\bigg)\diff{s},
\end{align}
which we can bound by
\begin{align}
\epsilon = \Big|\Tr(P\rho)\big|_0^1\Big| &\leq \Big|\Tr(\lambda^{-1}P'\rho)\big|_0^1\Big| + \int_0^1\bigg(\Tr\Big(\big|\frac{P^{\prime\prime}}{\lambda}\rho\big|\Big) + \Big|\Big(\frac{1}{\lambda}\Big)'\Big|\Tr\Big(\big|P^{\prime}\rho\big|\Big)\bigg)\diff{s} \\
&\leq \Big(\norm{\lambda(0)^{-1}P'(0)} + \norm{\lambda(1)^{-1}P'(1)}\Big)\Tr(\rho) + \int_0^1\bigg(\norm{\frac{P^{\prime\prime}}{\lambda}} + \Big|\Big(\frac{1}{\lambda}\Big)'\Big|\norm{P^{\prime}}\bigg)\Tr(\rho)\diff{s} \\
&\leq \norm{\lambda(0)^{-1}P'(0)} + \norm{\lambda(1)^{-1}P'(1)} + \int_0^1\bigg(\norm{\frac{P^{\prime\prime}}{\lambda}} + \Big|\Big(\frac{1}{\lambda}\Big)'\Big|\norm{P^{\prime}}\bigg)\diff{s}.
\end{align}
\end{proof}

The next step is to bound $\lVert P'\rVert$ and $\lVert P^{\prime\prime}\rVert$ by more useful quantities. We make use of the following lemma:

\begin{lemma} \label{lemma:projectorDerivativeBounds}
Under the assumptions stated in \ref{assumptions}, we have
\begin{enumerate}
\item $\lVert P'\rVert \leq 2 \frac{\lVert H'\rVert}{\Delta}$;
\item $\lVert P^{\prime\prime}\rVert \leq 8 \frac{\lVert H'\rVert^2}{\Delta^2} + 2 \frac{\lVert H^{\prime\prime}\rVert}{\Delta}$.
\end{enumerate}
\end{lemma}
This is fairly standard. See for example \cite{Reichardt}. A proof is provided in appendix \ref{appendix:projectorBounds}. We are now ready to use lemma \ref{lemma:errorBound} in two distinct contexts, leading to theorems \ref{theorem:constantRate} and \ref{theorem:adaptiveRate}.

\subsubsection{Constant $\lambda$}
We first derive a theorem under the assumption that $\lambda$ is constant.
In this case we obtain the following result:
\begin{theorem} \label{theorem:constantRate}
Under the assumptions in \ref{assumptions}, the algorithm \ref{procedure} produces the target state with fidelity of at least $1-\epsilon$ if $\lambda$ is constant and
\begin{equation}
\epsilon^{-1} 2\Big(\frac{\lVert H'(0)\rVert}{\Delta(0)} + \frac{\lVert H'(1)\rVert}{\Delta(1)} + \int_0^1 4 \frac{\lVert H'\rVert^2}{\Delta^2} + \frac{\lVert H^{\prime\prime}\rVert}{\Delta}\diff{s}\Big) \leq \lambda.
\end{equation}
In this case the time complexity of the procedure is given by $T = \lambda \int_0^1 \frac{1}{\Delta}\diff{s}$.
\end{theorem}
\begin{proof}
Let $\epsilon_0$ be the actual error of the algorithm. We need $\epsilon_0\leq \epsilon$. We can use lemma \ref{lemma:projectorDerivativeBounds} to rewrite the inequality in lemma \ref{lemma:errorBound}  as
\begin{align}
\epsilon_0 &\leq \lambda^{-1}\big(\norm{P'(0)} + \norm{P'(1)}\big) + \lambda^{-1}\int_0^1\norm{P^{\prime\prime}}\diff{s} \\
&\leq \lambda^{-1}\Big(2\frac{\lVert H'(0)\rVert}{\Delta(0)} + 2\frac{\lVert H'(1)\rVert}{\Delta(1)} + \int_0^1 8 \frac{\lVert H'\rVert^2}{\Delta^2} + 2 \frac{\lVert H^{\prime\prime}\rVert}{\Delta}\diff{s}\Big).
\end{align}
Set $B \defeq 2\Big(\frac{\lVert H'(0)\rVert}{\Delta(0)} + \frac{\lVert H'(1)\rVert}{\Delta(1)} + \int_0^1 4 \frac{\lVert H'\rVert^2}{\Delta^2} + \frac{\lVert H^{\prime\prime}\rVert}{\Delta}\diff{s}\Big)$.
Then we have
\begin{equation}
\epsilon_0 \leq \lambda^{-1}B \leq \epsilon B^{-1}B = \epsilon,
\end{equation}
so the algorithm works. The time complexity is simply given by $T = \int_0^1\frac{\lambda}{\Delta}\diff{s} = \lambda\int_0^1\frac{1}{\Delta}\diff{s}$.
\end{proof}
This result can be compared to theorem \ref{theorem:constantRateWithError} in the circuit model.

\subsubsection{Scaling $\lambda$ with the gap}
We know from \cite{RolandLocalAQC} and \cite{otherLocalAQC} that the performance of AQC can be improved with an adapted schedule, taking more time when the gap is small. Similarly we expect it to be possible to improve the performance of our procedure by varying $\lambda$. Indeed this is the case.

\begin{theorem} \label{theorem:adaptiveRate}
Under the assumptions in \ref{assumptions}, we additionally assume that there exists $0\leq q\leq 1$ and $B_1,B_2$ such that $\int_0^1 \frac{1}{\Delta^{1+q}}\diff{s} \leq B_1\Delta_m^{-q}$ and $\int_0^1 \frac{1}{\Delta^{2-q}}\diff{s} \leq B_2\Delta_m^{q-1}$ for all instances of the problem. Then algorithm \ref{procedure} produces the target state with a fidelity of at least $1-\epsilon$ if
\begin{equation}
\lambda = \epsilon^{-1}\frac{C}{\Delta^q\Delta_m^{1-q}},
\end{equation}
where $C\defeq 2\sup_{s\in[0,1]}\Big(2\norm{H'(s)} + 4\norm{H'(s)}^2B_2 + \norm{H''(s)} + q|\Delta'(s)|\,\norm{H'(s)}B_2 \Big)$.

In this case the time complexity of the procedure is given by
\[ T \leq \epsilon^{-1}\frac{B_1C}{\Delta_m}. \]
\end{theorem}
\begin{corollary} \label{corollary:adaptiveRate}
If $\int_0^1 \frac{1}{\Delta^p} \diff{s} = O(\Delta_m^{1-p})$ holds for all $p>1$, $|\Delta'| = O(1)$, $\norm{H'} = O(1)$ and $\norm{H''} = O(1)$, then algorithm \ref{procedure} with the rate defined in theorem \ref{theorem:adaptiveRate} produces the target state with fidelity $1-\epsilon$ and a time complexity of $O(\Delta_m^{-1})$ for all $0< q <1$.
\end{corollary}
\begin{proof}[Proof of theorem \ref{theorem:adaptiveRate}]
Let $\epsilon_0$ be the actual error of the algorithm, we need $\epsilon_0\leq \epsilon$. 
In this case the inequality in lemma \ref{lemma:errorBound} becomes
\begin{equation}
\epsilon_0 \leq \epsilon C^{-1}\Delta_m^{1-q}\big(\Delta(0)^{q}\norm{P'(0)} + \Delta(1)^{q}\norm{P'(1)}\big) + \epsilon C^{-1}\int_0^1\Delta^{q}\Delta_m^{1-q}\norm{P^{\prime\prime}} + \Big|\Big(\Delta^{q}\Delta_m^{1-q}\Big)'\Big|\norm{P^{\prime}}\diff{s}. \label{eq:errorBoundAdaptiveLambda}
\end{equation}
We bound the terms separately, using lemma \ref{lemma:projectorDerivativeBounds}. For the first, we have
\begin{equation}
\Delta_m^{1-q}\Delta^{q}\norm{P'} \leq 2\Delta_m^{1-q}\Delta^{q}\frac{\norm{H'}}{\Delta} = 2\Delta_m^{1-q}\frac{\norm{H'}}{\Delta^{1-q}} \leq 2\Delta_m^{1-q}\frac{\norm{H'}}{\Delta_m^{1-q}} = 2\norm{H'} \leq 2\sup_{s\in [0,1]}\norm{H'}
\end{equation}
at both $s=0$ and $s=1$, so we bound the sum by $4\sup_{s\in [0,1]}\norm{H'}$.

The second term splits into two, since we bound $\norm{P''}$ by $8\frac{\norm{H'}^2}{\Delta^2} + 2 \frac{\norm{H''}}{\Delta}$. For the first part we have
\begin{align}
8\int_0^1\Delta^{q}\Delta_m^{1-q}\frac{\norm{H^{\prime}}^2}{\Delta^2}\diff{s} &\leq 8\sup_{s\in[0,1]}\norm{H'(s)}^2\Delta_m^{1-q}\int_0^1\frac{1}{\Delta^{2-q}}\diff{s} \\
&\leq 8\sup_{s\in[0,1]}\norm{H'(s)}^2B_2\Delta_m^{1-q}\Delta_m^{q-1} = 8\sup_{s\in[0,1]}\norm{H'(s)}^2B_2.
\end{align}
The second part gives
\begin{align}
2\int_0^1\Delta^{q}\Delta_m^{1-q}\frac{\norm{H^{\prime\prime}}}{\Delta}\diff{s} &\leq 2\sup_{s\in[0,1]}\norm{H''(s)}\Delta_m^{1-q}\int_0^1\frac{1}{\Delta^{1-q}}\diff{s} \\
&\leq 2\sup_{s\in[0,1]}\norm{H''(s)}\Delta_m^{1-q}\Delta_m^{q-1} = 2\sup_{s\in[0,1]}\norm{H''(s)}.
\end{align}
Finally, for the third term,
\begin{align}
\int_0^1\Big|\Big(\Delta^{q}\Delta_m^{1-q}\Big)'\Big|\norm{P^{\prime}}\diff{s} &= \int_0^1 q\Delta^{q-1}\Delta_m^{1-q}\big|\Delta'\big|\norm{P^{\prime}}\diff{s} \\
&\leq 2q\Delta_m^{1-q}\Big(\sup_{s\in [0,1]}|\Delta'(s)|\,\norm{H'(s)}\Big)\int_0^1\frac{\Delta^{q-1}}{\Delta}\diff{s} \\
&= 2q\Delta_m^{1-q}\Big(\sup_{s\in [0,1]}|\Delta'(s)|\,\norm{H'(s)}\Big)\int_0^1\frac{1}{\Delta^{2-q}}\diff{s} \\
&\leq 2q\Big(\sup_{s\in [0,1]}|\Delta'(s)|\,\norm{H'(s)}\Big)\Delta_m^{1-q}B_2\Delta_m^{q-1} \\
&= 2qB_2\Big(\sup_{s\in [0,1]}|\Delta'(s)|\,\norm{H'(s)}\Big).
\end{align}
Plugging everything back into equation \eqref{eq:errorBoundAdaptiveLambda}, gives
\begin{align}
\epsilon_0 &\leq \epsilon C^{-1}\sup_{s\in[0,1]}\Big(4\norm{H'(s)} + 8\norm{H'(s)}^2B_2 + 2\norm{H''(s)} + 2q|\Delta'(s)|\,\norm{H'(s)}B_2 \Big) \\
&= \epsilon C^{-1}C = \epsilon,
\end{align}
so the procedure works. We can then calculate the time complexity
\begin{equation}
T = \int_0^1 \frac{\lambda}{\Delta} \diff{s} = \epsilon^{-1}\int_0^1 \frac{C}{\Delta^q\Delta_m^{1-q}\Delta} \diff{s} = \epsilon^{-1}C\Delta_m^{q-1}\int_0^1 \frac{1}{\Delta^{q+1}} \diff{s} \leq \epsilon^{-1}C\Delta_m^{q-1}B_1\Delta_m^{-q} = \epsilon^{-1}CB_1\Delta_m^{-1}.
\end{equation}
\end{proof}
These results can be compared to theorem \ref{theorem:adaptiveRateWithError} and corollary \ref{corollary:adaptiveRateWithError} in the circuit model.

\section{Improving the scaling in the error with eigenstate filtering}
The use of eigenstate filtering was introduced in \cite{eigenstateFiltering} to improve scaling in the error tolerance for algorithms based on adiabatic principles and the quantum Zeno effect, in particular with application to QLSP.

A similar technique was used in \cite{QLSPdiscreteAdiabaticTheorem} to achieve optimal scaling, but using Linear Combinations of Unitaries (LCU) instead of Quantum Signal Processing (QSP). We adapt the technique of \cite{QLSPdiscreteAdiabaticTheorem} to the present situation.

\begin{theorem} \label{eigenstateFiltering}
Let $H$ be a Hamiltonian with $\norm{H} \leq 1$ and $0$ in the spectrum of $H$, $\sigma(H)$. Suppose
\begin{itemize}
\item $\Delta\geq 0$ is such that $[- \Delta, \Delta]\cap \sigma(H) = \{0\}$;
\item $P$ is the orthogonal projector on the eigenspace associated to the eigenvalue $0$, we set $Q \defeq \mathbb{1}-P$;
\item $\rho$ is a density matrix of the form $P\rho_0 P + Q\rho_1 Q$ with $\Tr(P\rho_0) > 1/2$, that we can prepare at cost $T_0$;
\item $\epsilon > 0$.
\end{itemize}
Further, suppose
\begin{itemize}
\item we can adjoin two ancilla qubits to $\rho$;
\item we can can measure and reprepare the ancilla qubits;
\item we can evolve the system under $H\otimes R$ and $\mathbb{1}\otimes R$ for time $t$ for all Hermitian operators $R$ on $\mathbb{C}^{2\times 2}$ with $\norm{R} \leq 1$ at a cost of $t$.
\end{itemize}
Then we can prepare a state $\rho_2$ such that $\Tr(P\rho_2) \geq 1-\epsilon$ at a cost of $T = O(T_0+ \Delta^{-1}\log(1/\epsilon))$.
\end{theorem}
The idea of the procedure is relatively simple. With these assumptions, we can apply controlled versions of the unitary $e^{-i t H}$, i.e.\ $e^{it H\otimes \Pi}$ for some projector $\Pi$ on $\mathbb{C}^{2\times 2}$. This means that we can apply linear combinations of $e^{it H\otimes \Pi}$ using the technique of linear combinations of unitaries, see lemmas \ref{LCUlargeAncilla} and \ref{LCUsmallAncilla}. In particular we can apply a polynomial that has a large peak at $0$ and is very small everywhere else. We use this to filter out the part of the state that we do not want.

\begin{lemma}[LCU with arbitrarily large ancilla register] \label{LCUlargeAncilla}
Let $f(x) = \sum_{k=-n}^na_kx^k$ be a rational polynomial with complex coefficients such that $\sum_{k=-n}^n |a_k|^2 = 1$. Let $H$ be a Hamiltonian and $\rho$ the state of the system. Assume we have access to an ancilla register with orthonormal basis $\{\ket{k} \;|\; k\in \mathbb{Z}\}$. Then, at a cost of $O(nt)$, we can do an operation which either
\begin{itemize}
\item succeeds and applies $\sum_{k=-n}^n|a_k|^2e^{-itkH}$ to the system,
\item or fails, with a probability of $1 - \Tr\bigg(\Big(\sum_{k=-n}^n|a_k|^2e^{-itkH}\Big)\rho\Big(\sum_{k=-n}^n|a_k|^2e^{itkH}\Big)\bigg)$. We can see when this has happened thanks to the measured contents of the ancilla register.
\end{itemize}
\end{lemma}
\begin{proof}
The procedure is as follows: we first prepare the ancilla in the state $\ket{f} \defeq \sum_{k=-n}^n a_k\ket{k}$, then apply $\sum_{k=-n}^nkH\otimes \ketbra{k}{k}$ for time $t$ and finally measure the state $\ket{f}$. If we measure any other state than $\ket{f}$, the procedure fails.

The result then follows from the following identity:
\begin{equation}
\sum_{k,l = -n}^n\Big(\mathbb{1} \otimes a_k\bra{k}\Big)e^{-it\sum_{m=-n}^nmH\otimes \ketbra{m}{m}}\Big(\mathbb{1} \otimes \overline{a_l}\ket{l}\Big) = \sum_{m}|a_m|^2e^{-itmH}.
\end{equation}
Defining
\begin{equation}
\Pi_m^0 = \mathbb{1} - \sum_{k=0}^m\ketbra{k}{k} \qquad\text{and}\qquad \Pi_m^1 = \mathbb{1} - \sum_{k=-m}^{0}\ketbra{k}{k},
\end{equation}
we can write $e^{-it\sum_{m=-n}^nmH\otimes \ketbra{m}{m}} = \prod_{m=0}^{n-1} e^{-itH\otimes\Pi_m^0}e^{itH\otimes\Pi_m^1}$, which we can clearly apply at a cost of $2nt$.

The cost of ancilla preparation depends on the admissible operations on the ancilla register, but in a worst-case scenario, each $a_k$ needs to be set separately\footnote{This is the case for the procedure used in lemma \ref{LCUsmallAncilla}.} which means that the cost is $O(n)$. The total cost is then still $O(nt)$.
\end{proof}

\begin{lemma}[LCU with two ancilla qubits] \label{LCUsmallAncilla}
We can achieve the results of \ref{LCUlargeAncilla} only using two ancilla qubits at a time.
\end{lemma}
The construction is identical to the one in \cite{QLSPdiscreteAdiabaticTheorem}.

\begin{proof}[Proof of theorem \ref{eigenstateFiltering}]
Let $Q \defeq \mathbb{1}-P$ and write $Q = \sum_{j}Q_j$, where each $Q_j$ is an eigenprojector of $H$ associated to the eigenvalue $\omega_j$. Now we observe
\begin{equation}
\Big(\sum_{k=-n}^n|a_k|^2e^{-ikH}\Big)Q_j = \Big(\sum_{k=-n}^n|a_k|^2e^{-ik\omega_j}\Big)Q_j = A(\omega_j)Q_j,
\end{equation}
where $A(\omega)$ is the Fourier transform of the sequence $|a_k|^2$. Thus
\begin{align}
\Big(\sum_{k=-n}^n|a_k|^2e^{-ikH}\Big)Q\rho Q\Big(\sum_{k=-n}^n|a_k|^2e^{ikH}\Big) &= \sum_{j,l}\Big(\sum_{k=-n}^n|a_k|^2e^{-ikH}\Big)Q_j\rho Q_l\Big(\sum_{k=-n}^n|a_k|^2e^{ikH}\Big) \\
&= \sum_{j,l}A(\omega_j)A(-\omega_l)Q_j\rho Q_l.
\end{align}
Taking the trace gives $\Tr\Big(\sum_{j,l}A(\omega_j)A(-\omega_l)Q_j\rho Q_l\Big) \leq \max_{\omega \notin [-\Delta, \Delta]}A(\omega)^2 \Tr(Q\rho Q) \leq \max_{\omega \notin [-\Delta, \Delta]}A(\omega)^2$.
The goal then becomes to find a sequence and its Fourier transform such that $A(\omega_0) = 1$, $\max_{\omega \notin [-\Delta, \Delta]}A(\omega)^2 \leq \epsilon$ and whose window $n$ is as small as possible. The answer to this optimisation problem is well-known and is given by the Dolph-Chebyshev window \cite{DolphChebyshevWindow}. In this case we need a window of\footnote{We note that we improve the scaling by a factor of two compared to \cite{QLSPdiscreteAdiabaticTheorem}. This is because we are able to start from a state where $P\rho Q = 0 = Q\rho P$.}
\begin{align}
n = \frac{\cosh^{-1}(1/\sqrt{\epsilon})}{\cosh^{-1}\big(\sec(\Delta)\big)} \leq \frac{1}{2\Delta}\log\Big(\frac{4}{\epsilon}\Big).
\end{align}
By lemma \ref{LCUlargeAncilla}, we can implement this at a cost of $O(n)$. Note that this procedure terminates succesfully with a probability of at least $\Tr(P\rho_0)$ (which is bounded below) and we can check to see whether the procedure failed. If it failed, we repeat. On average we need to repeat fewer than $\Tr(P\rho_0)^{-1}$ times, which is $O(1)$.
\end{proof}

\section{Applications}
\subsection{Grover search}
For the Grover problem, we have an $N$-dimensional vector space we want to find an element of an $M$-dimensional subspace $\mathcal{M}$. In order to help us, we assume we have access to an oracle Hamiltonian $H_1 = \mathbb{1} - P_\mathcal{M}$, where $P_\mathcal{M}$ is the orthogonal projector on $\mathcal{M}$. In other words, we assume $H_1$ is admissible. We also assume $H_0 = \mathbb{1} - \ketbra{u}{u}$ is admissible, where $\ket{u} = \frac{1}{\sqrt{N}}\sum_{i=1}^N\ket{i}$ is the uniform superposition. The aim is now to use the interpolation $H(s) = (1-s)H_0 + sH_1$ to prepare as state in $\mathcal{M}$. For more details see \cite{farhiAQC} and \cite{Roland_2003}.

We see that $H(s)$ has four eigenvalues:
\begin{align}
\lambda_{1,2} &= \frac{1}{2}\left(1\pm \sqrt{1-4(1- \frac{M}{N})s(1-s)}\right) &\text{with multiplicity $1$} \\
\lambda_{3} &= 1-s &\text{with multiplicity $M-1$}\\
\lambda_{4} &= 1 &\text{with multiplicity $N-M-1$.}
\end{align}
The eigenvectors corresponding to $\lambda_3$ are the eigenvectors in $\mathcal{M}$ with zero overlap with $\ket{u}$. The eigenvectors corresponding to $\lambda_4$ are the eigenvectors in $\mathcal{M}^\perp$ with zero overlap with $\ket{u}$. Since the initial state has zero overlap with any of these vectors and they are eigenvectors of each $H(s)$, none of them are prepared by the procedure and everything happens in the two-dimensional space spanned by the eigenvectors associated to $\lambda_1$ and $\lambda_2$.

We have explicitly computed the gap, so we can use this as the bound $\Delta$:
\begin{equation}
\Delta(s) = \sqrt{1-4(1- \frac{M}{N})s(1-s)}. \label{eq:GroverGap}
\end{equation}
We can set $\Delta_m = \min_{s\in [0,1]} \Delta(s) = \sqrt{M/N}$. In order to give bounds on the time-complexity, we use the following result:
\begin{lemma} \label{lemma:GroverLemma}
For all $p > 1$ and $\Delta$ given by \eqref{eq:GroverGap}, we have
\begin{equation}
\int_0^1 \frac{1}{\Delta(s)^p}\diff{s} = O\big(\sqrt{N/M}^{p-1}\big) = O\big(\Delta_m^{1-p}\big),
\end{equation}
and, for $p=1$,
\begin{equation}
\int_0^1 \frac{1}{\Delta(s)}\diff{s} = O\big(\log(N/M)\big).
\end{equation}
\end{lemma}
We provide a proof in appendix \ref{appendix:GroverGap}.
For constant $\lambda$, we apply theorem \ref{theorem:constantRate} and use the lemma \ref{lemma:GroverLemma} to get a time complexity $O\big(\sqrt{N/M}\log(N/M)\big)$. 

We are able to take the $q$ in corollary \ref{corollary:adaptiveRate} to be anywhere in the range $0<q<1$, since for any such $q$ both $1+q$ and $2-q$ are strictly greater than $1$. This is related to the range of schedules described in \cite{An_QLSP}. The time complexity of the algorithm for any such $q$ is $O(\sqrt{N/M})$, since it is easy to check that to other conditions hold: $\norm{H'} = \norm{H_1 - H_0}$, $\norm{H''} = 0$ and
\begin{align}
|\Delta'| &= \Big|\frac{4(1 - \frac{M}{N})(\frac{1}{2}-s)}{\Delta}\Big| \\
&\leq \frac{2\sqrt{4(1 - \frac{M}{N})(\frac{1}{2}-s)^2}}{\Delta} \\
&\leq \frac{2\sqrt{\frac{M}{N} + 4(1 - \frac{M}{N})(\frac{1}{2}-s)^2}}{\Delta} = 2\frac{\Delta}{\Delta}  =2.
\end{align}

\subsection{Solving linear systems of equations}
The Quantum Linear Systems Problem (QLSP) was introduced in \cite{HHL}. Suppose $A$ is an invertible $N\times N$ matrix $b \in \mathbb{C}^N$ a vector. The goal is to prepare the quantum state $\frac{A^{-1}\ket{b}}{\norm{A^{-1}\ket{b}}}$. We express the time complexity of our algorithm in terms of the condition number $\kappa = \norm{A}\,\norm{A^{-1}}$.

We may restrict ourselves to Hermitian matrices because we can use the following trick from \cite{HHL}: If $A$ is not Hermitian, we consider the matrix $\begin{pmatrix}0 & A \\ A ^* & 0 \end{pmatrix}$, which has the same condition number, and solve the equation $\begin{pmatrix}0 & A \\ A ^* & 0 \end{pmatrix}\ket{y} = \begin{pmatrix}\ket{b} \\ 0\end{pmatrix}$.

First we rescale the matrix $A$ to $\frac{A}{\norm{A}}$. We do this because typically admissible matrices need to be uniformly bounded. This has the effect of shifting the lowest singular value from $\frac{1}{\norm{A^{-1}}}$ to $\frac{1}{\norm{A}\norm{A^{-1}}} = \kappa^{-1}$. Now we consider a path of Hamiltonians that was introduced in \cite{QLSPwithPR}. Define $A(s) \defeq (1-s)\sigma_z\otimes \mathbb{1} + s\sigma_x\otimes A$,  $Q_{b,+} \defeq \mathbb{1} - \big(\ket{+}\ket{b}\big)\big(\bra{+}\bra{b}\big)$ and $\sigma_{\pm} \defeq \frac{1}{2} \big(\sigma_x \pm i \sigma_y\big)$. Set $H(s) = \sigma_+\otimes \big(A(s)Q_{b,+}\big) + \sigma_-\otimes \big(Q_{b,+}A(s)\big)$. This can be written as a linear interpolation $H(s) = (1-s)H_0 + sH_1$, where
\begin{align}
H_0 &\defeq \sigma_+\otimes \big((\sigma_z\otimes \mathbb{1})Q_{b,+}\big) + \sigma_-\otimes \big(Q_{b,+}(\sigma_z\otimes \mathbb{1})\big) \\
H_1 &\defeq \sigma_+\otimes \big((\sigma_x\otimes A)Q_{b,+}\big) + \sigma_-\otimes \big(Q_{b,+}(\sigma_x\otimes A)\big).
\end{align}
Following the analysis of \cite{QLSPwithPR}, we see that $H(s)$ has $0$ as an eigenvalue for all $s\in [0,1]$. The corresponding eigenspace is spanned by $\{\ket{0}\otimes \ket{x(s)}, \ket{1}\otimes \ket{+}\ket{b}\}$, where $\ket{x(s)} \defeq \frac{A(s)^{-1}\ket{b}}{\norm{A(s)^{-1}\ket{b}}}$. Since $H(s)$ does not allow transition between these states, we are sure to not prepare $\ket{1}\otimes \ket{+}\ket{b}$, so long as we start with $\ket{0}\otimes \ket{x(0)}$.

In \cite{QLSPwithPR} it was also shown that the eigenvalue zero is separated from the rest of the spectrum by a gap that is at least
\begin{equation} \Delta(s) = \sqrt{(1-s)^2 + \Big(\frac{s}{\kappa}\Big)^2}. \label{LinSysGap} \end{equation}
If $\kappa$ is large enough, then we can take $\Delta_m \defeq \frac{1}{2\kappa} \leq \sqrt{\frac{1}{\kappa^2 + 1}} = \min_{s\in [0,1]}\Delta(s)$.

In order to give bounds on the time-complexity, we use the following result:
\begin{lemma} \label{lemma:QLSP}
For all $p > 1$, we have
\begin{equation}
\int_0^1 \frac{1}{\Delta(s)^p}\diff{s} = O\big(\kappa^{p-1}\big) = O\big(\Delta_m^{1-p}\big),
\end{equation}
and, for $p=1$,
\begin{equation}
\int_0^1 \frac{1}{\Delta(s)}\diff{s} = O\big(\log(\kappa)\big).
\end{equation}
\end{lemma}
We provide a proof in appendix \ref{appendix:QLSPgap}.

For constant $\lambda$, we apply theorem \ref{theorem:constantRate} and use the lemma \ref{lemma:QLSP} to get a time complexity $O\big(\kappa\log(\kappa)\big)$. This is also the complexity that was obtained in \cite{QLSPwithPR}.

As before, we have a full order reduction for $p>1$ and thus we are able to take the $q$ in corollary \ref{corollary:adaptiveRate} to be anywhere in the range $0<q<1$, since for any such $q$ both $1+q$ and $2-q$ are strictly greater than $1$. If $q=0$ or $q=1$, the complexity gains a factor of $\log(\kappa)$. This exactly mirrors the situation in \cite{An_QLSP} and is the reason why the algorithms for QLSP based on the RM have an extra factor of $\log(\kappa)$ in the asymptotic complexity, see \cite{QLSPwithPR} and \cite{QLSPwithPRnumerics}.

We can apply \ref{corollary:adaptiveRate} since $\norm{H'} = \norm{H_1 - H_0}$, $\norm{H''} = 0$ and
\begin{align}
|\Delta'| &= \Big|\frac{s-1 + s/\kappa^2}{\Delta}\Big| \\
&= \frac{\sqrt{(s-1 + s/\kappa^2)^2}}{\Delta} \\
&= \frac{\sqrt{(1+1/\kappa^2)^2s^2 -(1+1/\kappa^2)2s + 1}}{\Delta} \\
&\leq \frac{\sqrt{(1+1/\kappa^2)^2s^2 -(1+1/\kappa^2)2s + (1+1/\kappa^2)}}{\Delta} \\
&= \sqrt{1+1/\kappa^2}\frac{\Delta}{\Delta} = \sqrt{1+1/\kappa^2} = O(1).
\end{align}
This yields a time complexity of $O(\kappa)$ for fixed error tolerance. The scaling on both condition number and error tolerance is $O(\epsilon^{-1}\kappa)$. By a straightforward application of theorem \ref{eigenstateFiltering} at $s=1$, we get a scaling of $O\big(\log(\epsilon^{-1})\kappa\big)$. This is possible, since we know the eigenvalue of interest is $0$.

This result is optimal and matches the complexity reported in \cite{QLSPdiscreteAdiabaticTheorem}, where it was achieved using a very different method.

\vspace{1em}

This work was supported by the Belgian Fonds de la Recherche Scientifique - FNRS under Grants No. R.8015.21 (QOPT) and O.0013.22 (EoS CHEQS)

\printbibliography

\appendix

\section{Bounds on derivatives of projectors} \label{appendix:projectorBounds}
We provide a proof of lemma \ref{lemma:projectorDerivativeBounds}
\begin{lemma*}
Under the assumptions stated in \ref{assumptions}, we have
\begin{enumerate}
\item $\lVert P'\rVert \leq 2 \frac{\lVert H'\rVert}{\Delta}$;
\item $\lVert P^{\prime\prime}\rVert \leq 8 \frac{\lVert H'\rVert^2}{\Delta^2} + 2 \frac{\lVert H^{\prime\prime}\rVert}{\Delta}$.
\end{enumerate}
\end{lemma*}
\begin{proof}
Let $\Gamma$ be a circle in the complex plane, centred at the ground energy with radius $\Delta /2$. Then we have the Riesz form of the projector
\[ P = \frac{1}{2\pi i}\oint_\Gamma R_H(z)\diff{z}, \]
where $R_H(z) = \big(z\id - H\big)^{-1}$ is the resolvent of $H$ at $z$. Then $R_H(z)' = R_H(z)H'R_H(z)$ (the derivative is with respect to $s$, not $z$). As $H$ is a normal operator, the norm $\lVert R_H(z)\rVert$ is equal to the inverse of the distance from $z$ to the spectrum $\sigma(H)$. On the circle $\Gamma$ this is equal to $(\Delta/2)^{-1}$ everywhere. We can then approximate
\begin{align*}
\lVert P'\rVert &= \Big\lVert \frac{1}{2\pi i}\oint_\Gamma R_H(z)' \diff{z} \Big\rVert \\
&\leq  \frac{1}{2\pi}\oint_\Gamma \lVert R_H(z)'\lVert \diff{z} \\
&\leq \frac{1}{2\pi}\oint_\Gamma \lVert R_H(z) \rVert\cdot \lVert H' \rVert\cdot \lVert R_H(z)\rVert\diff{z} \\
&= \frac{1}{2\pi} \Big(\frac{2}{\Delta}\Big)^2 \lVert H'\rVert \oint_\Gamma \diff{z} \\
&= \frac{1}{2\pi} \Big(\frac{2}{\Delta}\Big)^2 2\pi \frac{\Delta}{2} \lVert H'\rVert \\
&= 2\frac{\lVert H'\rVert}{\Delta}.
\end{align*}
Similarly, we can write
\[ P^{\prime\prime} = \frac{1}{2\pi i}\oint_\Gamma 2R_H(z)H'R_H(z)H'R_H(z) + R_H(z)H^{\prime\prime}R_H(z)\diff{z}. \]
Estimating this in the same way as before yields
\[ \lVert P^{\prime\prime}\rVert \leq 8\frac{\lVert H'\rVert^2}{\Delta^2} + 2\frac{\lVert H^{\prime\prime}\rVert}{\Delta}.  \]
\end{proof}

\section{Comparison with the circuit model} \label{appendix:circuitModel}
So far we have assumed access to a device that can evolve a system under a given Hamiltonian in real time, i.e.\ it takes time $t$ to apply $e^{-itH}$. This is the typical setting of AQC and is also the setting of \cite{eigenpathTraversalPR} and \cite{QLSPwithPR}.

Many papers use a slightly different setup. In \cite{QLSPdiscreteAdiabaticTheorem} and \cite{QLSPwithPRnumerics} the setting is a standard quantum computer which is given access to block encodings of $H(s)$ for all $s\in [0,1]$. In this case the complexity is given by the number of times such a block-encoded Hamiltonian is used. In other words, the complexity is a query complexity rather than a time complexity.

Given access to only a block encoding of a Hamiltonian $H$, it is generally not possible to simulate $e^{-itH}$ exactly. Instead, we can use proposition \ref{optimalHamiltonianSimulation}, which is taken from \cite{QLSPwithPRnumerics}.
\begin{proposition}[Theorem 4 from \cite{QLSPwithPRnumerics}] \label{optimalHamiltonianSimulation}
Given access to an $(\alpha , m , 0)$–block-encoding $U_H$ of a Hermitian operator $H$ with $\norm{H} \leq 1$, we can realise a $(1, m + 2, \delta)$–block-encoding of $e^{-itH}$ for $t\in \mathbb{R}$ with
\[ 3\Big\lceil\frac{e}{2}\alpha |t| + \log\big(\frac{2c}{\delta}\big)\Big\rceil \]
calls to $U_H , U_H^*$ with $c = 4(\sqrt{2\pi}e^{\frac{1}{13}})^{-1} \approx 1.47762$.
\end{proposition}
In this proposition $\delta$ gives the error of the block encoding, i.e. if $U$ is the block encoding, then $\norm{U - e^{-itH}} \leq \delta$.

This motivates replacing the algorithm \ref{procedure} by the algorithm \ref{procedureWithError}, which now depends on both a rate $\lambda(s)$ and an allowable simulation error $\delta(s)$.

\begin{algorithm}
Pick a Poisson process $N: [0,1] \times (\Omega, \mathcal{A}, P)\to \mathbb{N}$ with rate $\lambda(s)$\;
At each jump point $s$ of the Poisson process, pick an instance $t$ of the random variable $T$ as defined in proposition \ref{phaseRandomisation} and use proposition \ref{optimalHamiltonianSimulation} to simulate the evolution under the time-independent Hamiltonian $H(s)$ for a time $t$ with error at most $\delta(s)$\;
\caption{Poisson-distributed phase randomisation with imperfect time evolution} \label{procedureWithError}
\end{algorithm}

In this case the number of queries is bounded by a quantity $Q$ is a random variable with stochastic differential equation $\diff{Q} = 3\Big(\frac{e}{2}\alpha |\tau| + \log\big(\frac{2c}{\delta}\big) + 1\Big)\diff{N}$. Taking the average yields $Q = 3\int_0^1\Big(\frac{e\alpha t_0}{2\Delta} + \log\big(\frac{2c}{\delta}\big) + 1\Big)\lambda\diff{s}$.

To analyse algorithm \ref{procedureWithError}, we prove lemma \ref{lemma:errorBoundWithError}, which is analogous to lemma \ref{lemma:errorBound}.

\begin{lemma} \label{lemma:errorBoundWithError}
Given the assumptions in \ref{assumptions} and that for all $s\in [0,1]$ and $t\in [0, \infty[$, we can apply an operation $A(s,t)$ such that $\norm{e^{-itH(s)} - A(s,t)} \leq \delta(s)$, the algorithm \ref{procedure} with rate $\lambda(s)$ has an error that is bounded by
\begin{equation}
\epsilon \leq \int_0^1 (2\delta + \delta^2)\Big(2\frac{\norm{H'}}{\Delta} + \lambda(s)\Big)\diff{s} + \norm{\lambda(0)^{-1}P'(0)} + \norm{\lambda(1)^{-1}P'(1)} + \int_0^1\bigg(\norm{\frac{P^{\prime\prime}}{\lambda}} + \Big|\Big(\frac{1}{\lambda}\Big)'\Big|\norm{P^{\prime}}\bigg)\diff{s}.
\end{equation}
\end{lemma}
\begin{proof}
The differential equation \eqref{eq:diffEq} is then
\begin{equation}
\rho' = \lambda\Big(P\rho P + Q\langle e^{-i\tau H}\rho e^{i\tau H} \rangle Q - \rho\Big) + \lambda\Big(A(s,t)\rho A(s,t)^* - e^{-i\tau H}\rho e^{i\tau H}\Big).
\end{equation}
We set $E \defeq A(s,t)\rho A(s,t)^* - e^{-i\tau H}\rho e^{i\tau H}$. Then equations \eqref{eq:firstPartPrho} and \eqref{eq:secondPartPrho} become $\Tr(P\rho') = \lambda \Tr(E)$ and $\Tr(P'\rho) = \Tr(P'E) - \Tr(\lambda^{-1}P'\rho')$, so
\begin{equation}
\Tr(P\rho)' = \Tr(P\rho') + \Tr(P'\rho) = \lambda \Tr(E) + \Tr(P'E) - \Tr(\lambda^{-1}P'\rho').
\end{equation}
The integral of $- \Tr(\lambda^{-1}P'\rho')$ was bounded in the proof of lemma \ref{lemma:errorBound}.

Using lemma \ref{simulationErrorLemma}, we bound $\Tr(|E|) \leq 2\delta + \delta^2$. Together with the bound $\norm{P'} \leq 2 \frac{\norm{H'}}{\Delta}$, this yields the result.
\end{proof}
In this proof we have made use of the following lemma:
\begin{lemma}[Lemma 9 from \cite{QLSPwithPRnumerics}] \label{simulationErrorLemma}
Suppose $A,B$ are operators such that $\norm{A-B} \leq \delta$. Then, for each density operators $\rho$, we have
\begin{equation}
\Tr(|A\rho A^* - B\rho B^*|) \leq 2\delta\norm{A} +\delta^2.
\end{equation}
\end{lemma}

\begin{theorem} \label{theorem:constantRateWithError}
Under the assumptions in \ref{assumptions}, the algorithm \ref{procedureWithError} produces the target state with fidelity of at least $1-\epsilon$ if $\lambda$ is constant, $\delta = \frac{4\epsilon}{27\lambda}$ and
\begin{equation}
\epsilon^{-1} 4\Big(\frac{\lVert H'(0)\rVert}{\Delta(0)} + \frac{\lVert H'(1)\rVert}{\Delta(1)} + \int_0^1 4 \frac{\lVert H'\rVert^2}{\Delta^2} + \frac{\lVert H^{\prime\prime}\rVert}{\Delta}\diff{s}\Big) \leq \lambda. \label{eq:constantLambdaBound}
\end{equation}
Using the Hamiltonian simulation of proposition \ref{optimalHamiltonianSimulation}, this gives a query complexity of
\[ Q = \lambda\Big(e\alpha t_0\frac{3}{2}\int_0^1\frac{1}{\Delta}\diff{s} + 3\log\big(\frac{27c}{2\epsilon}\big) + \log(\lambda)+1\Big). \]
\end{theorem}
\begin{proof}
Let $\epsilon_0$ be the actual error of the algorithm. We need $\epsilon_0\leq \epsilon$. As in the proof of \ref{theorem:constantRate}, rewrite the inequality in lemma \ref{lemma:errorBoundWithError} as
\begin{equation}
\epsilon_0 \leq (2\delta + \delta^2)\Big(\int_0^12\frac{\norm{H'}}{\Delta}\diff{s} + \lambda\Big) + \lambda^{-1}\Big(2\frac{\lVert H'(0)\rVert}{\Delta(0)} + 2\frac{\lVert H'(1)\rVert}{\Delta(1)} + \int_0^1 8 \frac{\lVert H'\rVert^2}{\Delta^2} + 2 \frac{\lVert H^{\prime\prime}\rVert}{\Delta}\diff{s}\Big).
\end{equation}
As in the proof of \ref{theorem:constantRate}, the second term is bounded by $\epsilon / 2$. Since all the terms of equation \eqref{eq:constantLambdaBound} are positive, we have $\int_0^12\frac{\norm{H'}}{\Delta}\diff{s} \leq \frac{\epsilon \lambda}{8}$. Then we bound
\begin{align}
(2\delta + \delta^2)\Big(\int_0^12\frac{\norm{H'}}{\Delta}\diff{s} + \lambda\Big) &\leq 3\delta\big( \frac{\epsilon \lambda}{8} + \lambda\big) \\
&\leq \frac{27}{8}\delta\lambda(\epsilon + 1) \leq \frac{27}{8}\delta\lambda \leq \frac{\epsilon}{2}.
\end{align}

Finally we consider the query complexity and calculate
\begin{align}
Q &= 3\int_0^1\Big(\frac{e\alpha t_0}{2\Delta} + \log\big(\frac{2c}{\delta}\big) + 1\Big)\lambda\diff{s} \\
&\leq \lambda\Big(e\alpha t_0\frac{3}{2}\int_0^1\frac{1}{\Delta}\diff{s} + 3\log\big(\frac{27c}{2\epsilon}\big) + \log(\lambda)+1\Big).
\end{align}
\end{proof}

\begin{theorem} \label{theorem:adaptiveRateWithError}
Under the assumptions in \ref{assumptions}, we additionally assume that there exists $0\leq q\leq 1$ and $B_1,B_2$ such that $\int_0^1 \frac{1}{\Delta^{1+q}}\diff{s} \leq B_1\Delta_m^{-q}$ and $\int_0^1 \frac{1}{\Delta^{2-q}}\diff{s} \leq B_2\Delta_m^{q-1}$ for all instances of the problem. Then algorithm \ref{procedureWithError} produces the target state with a fidelity of at least $1-\epsilon$ if
\begin{align}
\lambda &= \epsilon^{-1}\frac{2C}{\Delta^q\Delta_m^{1-q}} \\
\delta &= \frac{2\epsilon}{15\lambda}
\end{align}
where $C \defeq 2\sup_{s\in[0,1]}\Big(2\norm{H'(s)} + 4\norm{H'(s)}^2B_2 + \norm{H''(s)} + q|\Delta'(s)|\,\norm{H'(s)}B_2 \Big)$.

If, in addition, there exists a constant $B_3$ such that $\int_0^1 \frac{1}{\Delta^{2q}}\diff{s} \leq B_3\Delta_m^{1-2q}$ and the Hamiltonian simulation of \ref{optimalHamiltonianSimulation} is used, this gives a query complexity of
\[ Q \leq \frac{1}{\epsilon\Delta_m}\Big(12C\log(\epsilon^{-1}) + 3e\alpha t_0CB_1 + 6\log(15c)C + 12C^2B_3\Big). \]
\end{theorem}
\begin{proof}
As before, we need to bound the inequality in lemma \ref{lemma:errorBoundWithError}. Everything except the first term has already been bounded in the proof of theorem \ref{theorem:adaptiveRate} to be less than $\epsilon / 2$. (Notice that we are taking $\lambda$ to be twice the rate specified in theorem \ref{theorem:adaptiveRate}.)

We now need to show that the first term in the inequality in lemma \ref{lemma:errorBoundWithError} can be bounded by $\epsilon/2$. Indeed, we calculate
\begin{align}
\int_0^1 (2\delta + \delta^2)\Big(2\frac{\norm{H'}}{\Delta} + \lambda(s)\Big)\diff{s} &\leq \int_0^1 3\delta\Big(\frac{C}{2\Delta} + \lambda\Big)\diff{s} \\
&= \int_0^1 \frac{2\epsilon}{5\lambda}\Big(\frac{C}{2\Delta} + \lambda\Big)\diff{s} \\
&= \frac{2\epsilon}{5}\Big(1 + \int_0^1 \epsilon\frac{\Delta_m^{q-1}}{4\Delta^{q-1}}\diff{s}\Big) \\
&\leq \frac{2\epsilon}{5}\Big(1 + \frac{\epsilon}{4}\Big) \leq \frac{\epsilon}{2}.
\end{align}

Finally we consider the query complexity and calculate, using the fact that $\log(x)+1 \leq x$ for all positive $x$,
\begin{align}
Q &= 3\int_0^1\Big(\frac{e\alpha t_0}{2\Delta} + \log\big(\frac{2c}{\delta}\big) + 1\Big)\lambda\diff{s} \\
&= 3\int_0^1\Big(\frac{e\alpha t_0}{2\Delta} + \log\big(\frac{15c}{\epsilon}\big) + \log(\epsilon\lambda) + 1\Big)\lambda\diff{s} \\
&\leq 3\int_0^1\Big(\frac{e\alpha t_0}{2\Delta}\lambda + \log\big(\frac{15c}{\epsilon^2}\big)\lambda + \epsilon\lambda^2\Big)\diff{s}.
\end{align}
We bound each term separately. First
\begin{align}
3\int_0^1\frac{e\alpha t_0}{2\Delta}\lambda\diff{s} &= \frac{3e\alpha t_0C}{\epsilon\Delta_m^{1-q}}\int_0^1\frac{1}{\Delta^{q+1}}\lambda\diff{s} \\
&\leq \frac{3e\alpha t_0C}{\epsilon\Delta_m^{1-q}}B_1\Delta_m^{-q} = \frac{3e\alpha t_0CB_1}{\epsilon\Delta_m}.
\end{align}
Next
\begin{align}
3\int_0^1 \log\big(\frac{15c}{\epsilon^2}\big)\lambda \diff{s} &= 3\log\big(\frac{15c}{\epsilon^2}\big)\epsilon^{-1}\int_0^1\frac{2C}{\Delta^q\Delta_m^{1-q}}\diff{s} \\
&\leq \log\big(\frac{15c}{\epsilon^2}\big)\epsilon^{-1}\frac{6C}{\Delta_m}.
\end{align}
Finally
\begin{align}
3\int_0^1 \epsilon\lambda^2 \diff{s} &= \frac{12C^2}{\epsilon\Delta_m^{2-2q}}\int_0^1\frac{1}{\Delta^{2q}} \diff{s} \\
&\leq \frac{12C^2B_3}{\epsilon\Delta_m^{2-2q}}\Delta_m^{1-2q} = \frac{12C^2B_3}{\epsilon\Delta_m}.
\end{align}
Putting everything together yields the query complexity.
\end{proof}
\begin{corollary} \label{corollary:adaptiveRateWithError}
If $\int_0^1 \frac{1}{\Delta^p} \diff{s} = O(\Delta_m^{1-p})$ holds for all $p>1$, $|\Delta'| = O(1)$, $\norm{H'} = O(1)$ and $\norm{H''} = O(1)$, then algorithm \ref{procedureWithError} with the Hamiltonian simulation of \ref{optimalHamiltonianSimulation} and the parameters of \ref{theorem:adaptiveRateWithError} for some $1/2 <q <1$, produces a state with fidelity larger than $1-\epsilon$ using a number of queries that scales as $O\big(\Delta_m^{-1}\epsilon^{-1}\log(\epsilon^{-1})\big)$.
\end{corollary}
The asymptotic scaling in the error $\epsilon$ is slightly worse here, since there is an extra logarithmic factor, but this is not an issue if we want to apply eigenstate filtering. With eigenstate filtering the scaling in the error is still $O(\log(1/\epsilon))$.

\section{Gap properties}
\subsection{The gap in the Grover problem} \label{appendix:GroverGap}
For the Grover problem we have the following gap:
\begin{equation}
\Delta(s) = \sqrt{1-4(1- \frac{M}{N})s(1-s)}. \label{eq:GroverGapAppendix}
\end{equation}
We can set $\Delta_m = \min_{s\in [0,1]} \Delta(s) = \sqrt{M/N}$. We provide a proof of lemma \ref{lemma:GroverLemma}.

\begin{lemma*} 
For all $p > 1$ and $\Delta$ given by \eqref{eq:GroverGapAppendix}, we have
\begin{equation}
\int_0^1 \frac{1}{\Delta(s)^p}\diff{s} = O\big(\sqrt{N/M}^{p-1}\big) = O\big(\Delta_m^{1-p}\big),
\end{equation}
and, for $p=1$,
\begin{equation}
\int_0^1 \frac{1}{\Delta(s)}\diff{s} = O\big(\log(N/M)\big).
\end{equation}
\end{lemma*}
\begin{proof}
We note that $\Delta(s)$ is symmetric about $s= 1/2$. It is also strictly decreasing on $[0,1/2]$, going from $1$ to a minimum of $\sqrt{M/N}$.  So we can write
\begin{align}
\int_0^1 \frac{1}{\Delta(s)^p}\diff{s} &= 2\int_0^{1/2} \frac{1}{\Delta(s)^p}\diff{s} \\
&= 2\Big(\int_0^{1/2- \sqrt{M/N}} \frac{1}{\Delta(s)^p}\diff{s} + \int_{1/2- \sqrt{M/N}}^{1/2} \frac{1}{\Delta(s)^p}\diff{s} \Big).
\end{align}
Since $\Delta$ has a minimum of $\sqrt{M/N}$, we can bound the second integral by
\[ \int_{1/2- \sqrt{M/N}}^{1/2} \frac{1}{\Delta(s)^p}\diff{s} \leq \sqrt{\frac{M}{N}}\Big(\frac{1}{\min_{s\in[0,1]}\Delta(s)}\Big)^p = \frac{\sqrt{M/N}}{\sqrt{M/N}^p} = \sqrt{N/M}^{p-1}. \]
For the first integral, we write
\begin{align}
\int_0^{1/2 - \sqrt{M/N}} \frac{1}{\Delta(s)^p}\diff{s} &= \int_1^{\Delta\big(1/2- \sqrt{M/N}\big)} \frac{1}{\Delta^p}\dod{s}{\Delta}\diff{\Delta} \\
&= \int_{\Delta\big(1/2- \sqrt{M/N}\big)}^1 \frac{1}{\Delta^p}\Big(-\dod{s}{\Delta}\Big)\diff{\Delta}.
\end{align}
We can invert \eqref{eq:GroverGap} to obtain $s = \frac{1}{2} - \frac{1}{2}\sqrt{1-\frac{1-\Delta^2}{1-N/M}}$.
Then we have
\begin{equation}
-\od{s}{\Delta} = \frac{\Delta}{2\sqrt{(1-M/N)(\Delta^2 - M/N)}}.
\end{equation}
We now calculate
\[ \Delta\Big(\frac{1}{2} - \sqrt{\frac{M}{N}}\Big) = \sqrt{\frac{M}{N}}\sqrt{5 - 4 \frac{M}{N}} \geq 2\sqrt{\frac{M}{N}}, \]
assuming $M/N \leq 1/4$. So
\begin{align}
\int_0^{1/2 - \sqrt{M/N}} \frac{1}{\Delta^p}\diff{s} &\leq \int_{2\sqrt{\frac{M}{N}}}^1 \frac{1}{\Delta^p}\Big(-\dod{s}{\Delta}\Big)\diff{\Delta} \\
&= \int_{2\sqrt{\frac{M}{N}}}^1 \frac{1}{\Delta^p} \frac{\Delta}{2\sqrt{(1-M/N)(\Delta^2 - M/N)}}\diff{\Delta} \\
&\leq \int_{2\sqrt{\frac{M}{N}}}^1 \frac{1}{\Delta^p} \frac{\Delta}{2\sqrt{(1-M/N)(\Delta^2 - \Delta^2/4)}}\diff{\Delta} \\
&= \frac{1}{\sqrt{3(1-M/N)}} \int_{2\sqrt{\frac{M}{N}}}^1 \frac{1}{\Delta^p} \diff{\Delta}.
\end{align}
Now $\frac{1}{\sqrt{3(1-M/N)}}$ is $O(1)$ and $\int_{2\sqrt{\frac{M}{N}}}^1 \frac{1}{\Delta^p} \diff{\Delta} = \Big[\frac{1}{(p-1)\Delta^{p-1}}\Big]_{2\sqrt{M/N}}^1$ is $O\big(\sqrt{N/M}^{p-1}\big)$, if $p>1$. If $p=1$, then it is $O\big(\log\sqrt{N/M}\big)$.
\end{proof}

\subsection{The gap in QLSP} \label{appendix:QLSPgap}
For the quantum linear system problem we have the following bound on the gap:
\begin{equation} \Delta(s) = \sqrt{(1-s)^2 + \Big(\frac{s}{\kappa}\Big)^2}. \label{LinSysGapAppendix} \end{equation}
If $\kappa$ is large enough, then we can take $\Delta_m \defeq \frac{1}{2\kappa} \leq \sqrt{\frac{1}{\kappa^2 + 1}} = \min_{s\in [0,1]}\Delta(s)$. We provide a proof of lemma \ref{lemma:QLSP}.
\begin{lemma*}
For all $p > 1$, we have
\begin{equation}
\int_0^1 \frac{1}{\Delta(s)^p}\diff{s} = O\big(\kappa^{p-1}\big) = O\big(\Delta_m^{1-p}\big),
\end{equation}
and, for $p=1$,
\begin{equation}
\int_0^1 \frac{1}{\Delta(s)}\diff{s} = O\big(\log(\kappa)\big).
\end{equation}
\end{lemma*}
\begin{proof}
We note that $\Delta(s)$ is strictly decreasing on $\Big[0,1- \frac{1}{\kappa^2 + 1}\Big]$, going from $1$ to a minimum of $\sqrt{\frac{1}{\kappa^2 + 1}}$.  So we can write
\[ \int_0^1 \frac{1}{\Delta(s)^p}\diff{s} = \int_0^{1- \frac{1}{\kappa^2 + 1}} \frac{1}{\Delta(s)^p}\diff{s} + \int_{1- \frac{1}{\kappa^2 + 1}}^1 \frac{1}{\Delta(s)^p}\diff{s}. \]
Since $\Delta$ has a minimum of $\sqrt{\frac{1}{\kappa^2 + 1}}$, we can bound the second integral by
\[ \int_{1- \frac{1}{\kappa^2 + 1}}^{1} \frac{1}{\Delta(s)^p}\diff{s} \leq \frac{1}{\kappa^2 + 1}\Big(\frac{1}{\min_{s\in[0,1]}\Delta(s)}\Big)^p = \frac{1}{\kappa^2 + 1}\big(\kappa^2 + 1\big)^{p/2} = \big(\kappa^2 + 1\big)^{p/2-1}. \]
For the first integral, we write
\begin{align}
\int_0^{1 - \frac{1}{\kappa^2+1}} \frac{1}{\Delta^p}\diff{s} &= \int_1^{\Delta\big(1 - \frac{1}{\kappa^2+1}\big)} \frac{1}{\Delta^p}\dod{s}{\Delta}\diff{\Delta} \\
&= \int_{\Delta\big(1 - \frac{1}{\kappa^2+1}\big)}^1 \frac{1}{\Delta^p}\Big(-\dod{s}{\Delta}\Big)\diff{\Delta} \\
&= \int_{\sqrt{\frac{1}{\kappa^2+1}}}^1 \frac{1}{\Delta^p}\Big(-\dod{s}{\Delta}\Big)\diff{\Delta}.
\end{align}
We can invert \eqref{LinSysGapAppendix} on $\Big[0,1- \frac{1}{\kappa^2 + 1}\Big]$ to obtain $s = \frac{\kappa^2}{\kappa^2+1}(1-\Delta)$.
Then we have
\begin{equation}
-\od{s}{\Delta} = \frac{\kappa^2}{\kappa^2+1},
\end{equation}
so
\begin{align}
\int_0^{1 - \frac{1}{\kappa^2+1}} \frac{1}{\Delta^p}\diff{s} &= \int_{\sqrt{\frac{1}{\kappa^2+1}}}^1 \frac{1}{\Delta^p}\frac{\kappa^2}{\kappa^2+1}\diff{\Delta} \\
&= \frac{\kappa^2}{\kappa^2+1}\Big(\frac{1}{(p-1)\Delta^{p-1}}\Big)\Big|^{\Delta = \sqrt{\frac{1}{\kappa^2+1}}}_{\Delta = 1} \\
&= O(\kappa^{p-1}).
\end{align}
If $p=1$, then the integral is $O\big(\log(\kappa)\big)$.
\end{proof}

\end{document}